\documentclass{lmcs}


\usepackage{mathtools}
\usepackage{tikz-cd}
\usepackage{xcolor}

\newcommand{\Id}{\mathrm{Id}}
\newcommand{\refl}{\rho}
\newcommand{\symm}{\sigma}
\newcommand{\trans}{\tau}
\newcommand{\Rw}{\rightsquigarrow}
\newcommand{\RwEq}{\sim}
\newcommand{\Path}{\mathrm{Path}}
\newcommand{\Cell}{\mathrm{Cell}}
\newcommand{\src}{\mathrm{src}}
\newcommand{\tgt}{\mathrm{tgt}}
\newcommand{\comp}{\circ}
\newcommand{\inv}{\mathrm{inv}}
\newcommand{\id}{\mathrm{id}}
\newcommand{\assoc}{\alpha}
\newcommand{\lunit}{\lambda}
\newcommand{\runit}{\varrho}
\newcommand{\linv}[1]{\iota_{L,#1}}
\newcommand{\rinv}[1]{\iota_{R,#1}}

\begin{document}

\title[Computational Paths Form a Weak $\omega$-Groupoid]{Computational Paths Form a Weak \texorpdfstring{$\omega$}{omega}-Groupoid}

\author[A.~F.~Ramos]{Arthur F. Ramos}
\address{Microsoft, USA}
\email{arfreita@microsoft.com}

\author[T.~M.~L.~de~Veras]{Tiago M. L. de Veras}
\address{Departamento de Matem\'atica, Universidade Federal Rural de Pernambuco, Brazil}
\email{tiago.veras@ufrpe.br}

\author[R.~J.~G.~B.~de~Queiroz]{Ruy J. G. B. de Queiroz}
\address{Centro de Inform\'atica, Universidade Federal de Pernambuco, Brazil}
\email{ruy@cin.ufpe.br}

\author[A.~G.~de~Oliveira]{Anjolina G. de Oliveira}
\address{Centro de Inform\'atica, Universidade Federal de Pernambuco, Brazil}
\email{ago@cin.ufpe.br}

\begin{abstract}
Lumsdaine (2010) and van den Berg-Garner (2011) proved that types in Martin-L\"of type theory carry the structure of weak omega-groupoids. Their proofs, while foundational, rely on abstract properties of the identity type without providing explicit computational content for coherence witnesses.

We establish an analogous result for computational paths---an alternative formulation of equality where witnesses are explicit sequences of rewrites from the LND-EQ-TRS term rewriting system. Our main result is that computational paths on any type form a weak omega-groupoid with fully explicit coherence data. The groupoid operations---identity, composition, and inverse---are defined at every dimension, and the coherence laws (associativity, unit laws, inverse laws) are witnessed by concrete rewrite derivations rather than abstract existence proofs.

The construction provides: (i) a proper tower of n-cells for all dimensions, with 2-cells as derivations between paths and higher cells mediating between lower-dimensional witnesses; (ii) explicit pentagon and triangle coherences built from the rewrite rules; and (iii) contractibility at dimensions at least 3, ensuring all parallel higher cells are connected. The contractibility property is derived from the normalization algorithm of the rewrite system, grounding the higher-dimensional structure in concrete computational content.

The entire construction has been formalized in Lean 4, providing machine-checked verification of the weak omega-groupoid structure.
\end{abstract}

\maketitle

\section{Introduction}
\label{sec:introduction}

The relationship between Martin-L\"of's intensional type theory and higher categorical structures has been a central topic in the foundations of mathematics since the discovery of homotopy type theory. The seminal work of Hofmann and Streicher \cite{HofmannStreicher1994} demonstrated that types in intensional type theory cannot satisfy uniqueness of identity proofs (UIP), revealing that the identity type has a rich structure beyond mere reflexivity.\footnote{The significance of this insight cannot be overstated. In a 2015 email exchange, Vladimir Voevodsky---who later developed the univalence axiom---reflected on his initial reaction to identity types: ``I was not inspired by [Hofmann-Streicher]. In fact I tried several times to understand what they are saying and never could. [...] All the time before it I was hypnotized by the mantra that the only inhabitant of the Id type is reflexivity which made [them] useless from my point of view.'' This ``mantra'' had obscured the rich higher structure of identity types for years, even from leading mathematicians. See \url{https://groups.google.com/g/homotopytypetheory/c/K_4bAZEDRvE/m/VbYjok7bAAAJ}.} This observation was dramatically extended by Lumsdaine \cite{Lumsdaine2009} and van den Berg and Garner \cite{BergGarner2011}, who independently proved that types in Martin-L\"of type theory form weak $\omega$-groupoids.

These results are profound but abstract: they show the existence of higher groupoid structure without providing explicit computational content for the witnesses of coherence laws. In contrast, the theory of \emph{computational paths} \cite{Queiroz2016Paths, Ramos2017IdentityPaths, Ramos2018ExplicitPaths} provides an alternative formulation of the identity type where witnesses of propositional equality are explicit sequences of rewrites---computational paths that transform one term into another through applications of definitional equality rules.

In our previous work \cite{Veras2023WeakGroupoid}, we established that computational paths form a weak groupoid (a 1-groupoid where laws hold up to higher equivalence rather than strict equality). The present paper extends this result to its natural conclusion: \emph{computational paths form a weak $\omega$-groupoid}, the full higher categorical structure expected from the identity type interpretation.

\subsection{Main Contributions}

This paper provides:

\begin{enumerate}
    \item A complete, self-contained pen-and-paper proof that computational paths on any type $A$ form a weak $\omega$-groupoid, following the Lumsdaine/van den Berg-Garner construction.

    \item A proper tower of $n$-cells: $\Cell_2$ (2-cells, i.e., derivations between paths), $\Cell_3$ (3-cells between 2-cells), $\Cell_4$ (4-cells between 3-cells), and $\Cell_n$ for arbitrary dimensions $n$.

    \item The \emph{canonicity axiom}: every derivation is connected to a canonical derivation $\gamma_{p,q}$ that factors through normal forms. Contractibility at all dimensions $\geq 3$ is \emph{derived} from this axiom. The construction is semantically justified by normalization and confluence of the rewrite system \cite{Ruy4, Ramos2018ExplicitPaths}, grounding the higher-dimensional structure in concrete computational content.

    \item Explicit pentagon and triangle coherence witnesses constructed from the $\mathrm{LND}_{\mathrm{EQ}}$-TRS rewrite rules.

    \item Groupoid operations (identity, composition, inverse) at each level with proper source/target behavior.
\end{enumerate}

\subsection{Related Work}

The computational paths approach to the identity type originates in work by de Queiroz and colleagues \cite{Queiroz2016Paths}, building on earlier ideas about equality proofs as sequences of rewrites \cite{Ruy4}. The formalization connects to term rewriting systems \cite{BaaderNipkow1998, KnuthBendix1970}, particularly the $\mathrm{LND}_{\mathrm{EQ}}$-TRS system \cite{Ramos2018Thesis} that captures all rewrite rules between computational paths. The connection between rewriting and coherence has been developed by Kraus and von Raumer \cite{KrausVonRaumer2022}, who showed that confluence and termination suffice to build coherent higher-dimensional structure, and by Mimram \cite{Mimram2023Coherence} for categorical coherence.

On the categorical side, weak $\omega$-groupoids were characterized by Batanin, Leinster \cite{leinster1}, and others. The connection to type theory was established through the groupoid model \cite{HofmannStreicher1994} and extended to higher dimensions in \cite{Lumsdaine2009, BergGarner2011}. Awodey's work \cite{Awodey2012} provides essential background on type theory and homotopy. More recently, Finster and Mimram \cite{FinsterMimram2017} developed CaTT, a type theory whose models are exactly weak $\omega$-categories, with Benjamin et al.\ \cite{BenjaminFinsterMimram2021} proving the initiality conjecture for this system.

The computational content of homotopy type theory has been explored through cubical type theory \cite{Cohen2018Cubical}, which provides a constructive interpretation of the univalence axiom. Recent advances include the formalization of $\pi_4(\mathbb{S}^3)$ in Cubical Agda \cite{LjungstromMortberg2023} and two-level type theory \cite{AnnenkovCapriottiKrausSattler2023}, which allows mixing strict and weak equality. These developments demonstrate the increasing computational sophistication of higher-dimensional type theory.

The present work complements the Lean 4 formalization of computational paths available at \cite{ComputationalPathsLean}, which provides a verified implementation of the constructions described herein.

\subsection{Structure of the Paper}

Section~\ref{sec:background} reviews the necessary background on type theory, computational paths, and higher category theory. Section~\ref{sec:cells} defines the cell structure at each dimension. Section~\ref{sec:operations} introduces the groupoid operations (identity, composition, inverse) and proves they satisfy the required source/target conditions. Section~\ref{sec:globular} establishes the globular identities. Section~\ref{sec:coherence} proves the coherence laws, including associativity, unit laws, and inverse laws. Section~\ref{sec:higher} addresses the pentagon and triangle coherences, as well as the contractibility of higher cells. Section~\ref{sec:main} assembles these results into the main theorem. Section~\ref{sec:conclusion} discusses implications and future work.

\section{Background}
\label{sec:background}

\subsection{Martin-L\"of Type Theory and the Identity Type}

In Martin-L\"of's intensional type theory \cite{MartinLoef1984}, the \emph{identity type} $\Id_A(a,b)$ is formed for any type $A$ and terms $a, b : A$. The traditional formulation provides:

\begin{itemize}
    \item \textbf{Formation}: Given $A$ type and $a, b : A$, we form $\Id_A(a,b)$ type.
    \item \textbf{Introduction}: For any $a : A$, we have $\refl(a) : \Id_A(a,a)$.
    \item \textbf{Elimination (J-rule)}: The recursor $J$ allowing proofs by path induction.
\end{itemize}

The $J$-eliminator states that to prove a property $C(x, y, p)$ for all $x, y : A$ and $p : \Id_A(x,y)$, it suffices to prove $C(x, x, \refl(x))$ for all $x : A$. While elegant, this formulation makes the identity type opaque: the only canonical proof is reflexivity, and other proofs must be constructed via the complex $J$-eliminator without revealing the computational steps that justify the equality.

A natural question arises from the Brouwer-Heyting-Kolmogorov interpretation of logic: \emph{what constitutes a proof of an equality statement $t_1 = t_2$?} The BHK interpretation provides answers for conjunction (a pair of proofs), disjunction (a proof of either disjunct), implication (a function transforming proofs), and quantifiers, but leaves equality unspecified. The computational paths approach fills this gap.

\subsection{The Definitional Equality of Type Theory}

Before defining computational paths, we must understand the equality theory they are built upon. In type theory, \emph{definitional equality} (also called judgmental or computational equality) is the relation that determines when two terms are considered identical by the type checker. For the $\Pi$-type (dependent function type), the standard rules include:

\begin{itemize}
    \item \textbf{$\beta$-reduction}: $(\lambda x.M)N = M[N/x]$ --- function application
    \item \textbf{$\eta$-expansion}: $(\lambda x.Mx) = M$ when $x \notin FV(M)$ --- extensionality
    \item \textbf{$\xi$-rule}: From $M = M'$, derive $\lambda x.M = \lambda x.M'$ --- congruence under $\lambda$
    \item \textbf{$\mu$-rule}: From $M = M'$, derive $NM = NM'$ --- congruence in argument
    \item \textbf{$\nu$-rule}: From $M = M'$, derive $MN = M'N$ --- congruence in function
    \item \textbf{$\rho$-rule}: $M = M$ --- reflexivity
    \item \textbf{$\sigma$-rule}: From $M = N$, derive $N = M$ --- symmetry
    \item \textbf{$\tau$-rule}: From $M = N$ and $N = P$, derive $M = P$ --- transitivity
\end{itemize}

Similar rules exist for other type formers ($\Sigma$-types, sum types, etc.). The key observation is that two terms $M$ and $N$ are definitionally equal if and only if there exists a \emph{sequence of rule applications} transforming $M$ into $N$. This sequence is precisely what we call a computational path.

\subsection{Computational Paths: Formal Definition}

\begin{defi}[Computational Path {\cite{Queiroz2016Paths, Ramos2017IdentityPaths}}]
Let $a$ and $b$ be terms of type $A$. A \emph{computational path} $s$ from $a$ to $b$, written $a =_s b : A$, is a composition of rewrites where each rewrite is an application of one of the inference rules of the definitional equality theory. The path $s$ is a term that \emph{records} the sequence of rules applied.
\end{defi}

\begin{exa}
Consider the term $M \equiv (\lambda x.(\lambda y.yx)(\lambda w.zw))v$ in the untyped $\lambda$-calculus. We can show $M =_{\beta\eta} zv$ via the sequence:
\begin{enumerate}
    \item $(\lambda x.(\lambda y.yx)(\lambda w.zw))v$ --- original term
    \item $(\lambda x.(\lambda y.yx)z)v$ --- by $\eta$-reduction on $(\lambda w.zw)$
    \item $(\lambda y.yv)z$ --- by $\beta$-reduction
    \item $zv$ --- by $\beta$-reduction
\end{enumerate}
The computational path recording this derivation is:
\[
s = \trans(\eta, \trans(\beta, \beta))
\]
where each step names the rule applied and $\trans$ composes consecutive steps.
\end{exa}

The fundamental operations on computational paths are:

\begin{itemize}
    \item \textbf{Reflexivity} $\refl$: For any $a : A$, we have $a =_\refl a : A$. This is the trivial path of zero steps.

    \item \textbf{Symmetry} $\symm$: From $a =_s b : A$, derive $b =_{\symm(s)} a : A$. This reverses the direction of each step in $s$.

    \item \textbf{Transitivity} $\trans$: From $a =_s b : A$ and $b =_t c : A$, derive $a =_{\trans(s,t)} c : A$. This concatenates the two sequences of rewrites.
\end{itemize}

Additionally, computational paths support \emph{congruence} operations that lift paths through term constructors:

\begin{itemize}
    \item \textbf{$\mu_f$}: From $a =_s b : A$ and $f : A \to B$, derive $f(a) =_{\mu_f(s)} f(b) : B$.
    \item \textbf{$\xi$}: From $M =_s N : B$ (with $x : A$ free), derive $\lambda x.M =_{\xi(s)} \lambda x.N : A \to B$.
\end{itemize}

\subsection{The Path-Based Identity Type}

Using computational paths, we reformulate the identity type with explicit witnesses:

\[
\dfrac{A \text{ type} \quad a : A \quad b : A}{\Id_A(a,b) \text{ type}} \; (\Id\text{-F})
\qquad
\dfrac{a =_s b : A}{s(a,b) : \Id_A(a,b)} \; (\Id\text{-I})
\]

The introduction rule states that any computational path $s$ from $a$ to $b$ gives rise to a term $s(a,b)$ inhabiting the identity type. This makes the \emph{reason} for equality explicit: the path $s$ records exactly which rules were applied to transform $a$ into $b$.

The elimination rule uses a constructor $\mathrm{REWR}$ (for ``rewrite''):
\[
\dfrac{m : \Id_A(a,b) \quad [a =_g b : A] \vdash h(g) : C}{\mathrm{REWR}(m, \acute{g}.h(g)) : C} \; (\Id\text{-E})
\]

This says: if from an arbitrary path $g$ witnessing $a = b$ we can construct $h(g) : C$, and we have an actual witness $m : \Id_A(a,b)$, then we can construct a term of type $C$ by substituting the path underlying $m$ for $g$ in $h$.

\subsection{The $\mathrm{LND}_{\mathrm{EQ}}$-TRS: A Term Rewriting System for Paths}

A crucial observation is that \emph{different paths between the same endpoints may be equivalent}. For instance, applying symmetry twice returns to the original path: $\symm(\symm(s))$ should be equivalent to $s$. Similarly, composing a path with the reflexive path should yield the original: $\trans(s, \refl) \sim s$.

These equivalences are captured by the $\mathrm{LND}_{\mathrm{EQ}}$-TRS (Labelled Natural Deduction Equality Term Rewriting System), originating in the work of de Queiroz \cite{Ruy4} and further developed in \cite{Ramos2018Thesis}. This system provides a complete set of rewrite rules for computational paths, organized into several families:

\subsubsection{Basic Reductions}

The first family handles interactions between the fundamental operations:

\begin{align}
\symm(\refl) &\Rw \refl \tag{$\rhd_{sr}$} \\
\symm(\symm(r)) &\Rw r \tag{$\rhd_{ss}$} \\
\trans(r, \symm(r)) &\Rw \refl \tag{$\rhd_{tr}$} \\
\trans(\symm(r), r) &\Rw \refl \tag{$\rhd_{tsr}$} \\
\trans(r, \refl) &\Rw r \tag{$\rhd_{trr}$} \\
\trans(\refl, r) &\Rw r \tag{$\rhd_{tlr}$}
\end{align}

Rule $\rhd_{sr}$ says the symmetry of reflexivity is reflexivity. Rule $\rhd_{ss}$ says double symmetry cancels. Rules $\rhd_{tr}$ and $\rhd_{tsr}$ say a path composed with its inverse yields reflexivity. Rules $\rhd_{trr}$ and $\rhd_{tlr}$ are the unit laws for transitivity.

\subsubsection{Associativity}

The associativity of transitivity is captured by:
\begin{equation}
\trans(\trans(r, s), t) \Rw \trans(r, \trans(s, t)) \tag{$\rhd_{tt}$}
\end{equation}

This rule is fundamental for the groupoid structure: it says that the order of composing three paths doesn't matter (up to rewrite equivalence).

\subsubsection{Distributivity of Symmetry}

Symmetry distributes over transitivity with order reversal:
\begin{equation}
\symm(\trans(r, s)) \Rw \trans(\symm(s), \symm(r)) \tag{$\rhd_{stss}$}
\end{equation}

This reflects the fact that reversing a composite path requires reversing each component and their order.

\subsubsection{Congruence Rules}

The system includes rules for how congruence operations interact with the basic operations. For instance:
\begin{align}
\mu_f(\trans(p, q)) &\Rw \trans(\mu_f(p), \mu_f(q)) \tag{$\rhd_{tf}$} \\
\mu_f(\symm(p)) &\Rw \symm(\mu_f(p)) \tag{$\rhd_{sm}$} \\
\mu_g(\mu_f(p)) &\Rw \mu_{g \circ f}(p) \tag{$\rhd_{cf}$} \\
\mu_{\mathrm{id}}(p) &\Rw p \tag{$\rhd_{ci}$}
\end{align}

These say that applying a function commutes with path operations, and that function composition corresponds to composition of congruence operations.

\subsubsection{$\beta\eta$-Reductions for Paths}

The system also includes rules that capture interactions between path operations and the type-theoretic $\beta\eta$ rules. For example:
\begin{equation}
\nu(\xi(r)) \Rw r \tag{$\rhd_{mxl}$}
\end{equation}

This says that lifting a path under $\lambda$ and then applying it ($\nu$) recovers the original path.

The full $\mathrm{LND}_{\mathrm{EQ}}$-TRS contains over 40 rules covering all interactions between path operations, congruences for all type formers, and substitution operations. We refer to \cite{Ramos2018Thesis} for the complete system.

\subsection{Rewrite Equivalence}

\begin{defi}[Rewrite Step Relation]
A \emph{rewrite step} $p \Rw q$ holds when $q$ is obtained from $p$ by applying one of the $\mathrm{LND}_{\mathrm{EQ}}$-TRS rules at some position in $p$.
\end{defi}

\begin{defi}[Step Type]
\label{def:step}
For paths $p, q : \Path(a, b)$, the type $\mathrm{Step}(p, q)$ is defined as the type of witnesses that $p \Rw q$ in a single rewrite step. Formally:
\[
\mathrm{Step}(p, q) := \{ s \mid s \text{ witnesses } p \Rw q \text{ by a single rule application} \}
\]
Crucially, $\mathrm{Step}$ is \emph{propositional}: for any $p, q$, the type $\mathrm{Step}(p, q)$ has at most one inhabitant. This reflects that we care only whether a rewrite step exists, not which rule was applied.
\end{defi}

\begin{defi}[Rewrite Closure]
The \emph{rewrite closure} $\mathrm{Rw}$ is the reflexive-transitive closure of the rewrite step relation.
\end{defi}

\begin{defi}[Rewrite Equivalence]
The \emph{rewrite equivalence} relation $\RwEq$ on $\Path(a,b)$ is the equivalence closure of $\mathrm{Rw}$:
\[
p \RwEq q \iff \exists r_1, \ldots, r_n : p = r_1, \; r_n = q, \; \forall i : r_i \Rw r_{i+1} \text{ or } r_{i+1} \Rw r_i
\]
\end{defi}

\begin{thm}[Properties of $\RwEq$ {\cite[Chapter 4]{Ramos2018Thesis}}]
The rewrite equivalence relation satisfies:
\begin{enumerate}
    \item \textbf{Reflexivity}: $p \RwEq p$ for all paths $p$.
    \item \textbf{Symmetry}: If $p \RwEq q$ then $q \RwEq p$.
    \item \textbf{Transitivity}: If $p \RwEq q$ and $q \RwEq r$ then $p \RwEq r$.
    \item \textbf{Congruence}: If $p \RwEq p'$ and $q \RwEq q'$ then $\trans(p,q) \RwEq \trans(p',q')$, and similarly for $\symm$ and congruence operations.
\end{enumerate}
\end{thm}

\begin{rem}[Proof Irrelevance and Step Equality]
\label{rem:proof-irrel}
Crucially, in our formalization $\RwEq$ is defined as a \emph{propositional} relation: for any paths $p, q : \Path(a,b)$, the type ``$p \RwEq q$'' is a proposition (has at most one inhabitant). This means that if $p \RwEq q$ holds, there is a unique witness of this fact.

This proof-irrelevance has two important consequences:
\begin{enumerate}
    \item It provides semantic justification for the canonicity axiom: since Step is proof-irrelevant, any two derivations between the same endpoints represent ``the same'' computational witness and should connect to the same canonical derivation.
    \item It ensures that the step equality axiom ($\mathsf{step\_eq}$) holds: any two steps with the same source and target are equal.
\end{enumerate}
Combined with the normalization and confluence of the $\mathrm{LND}_{\mathrm{EQ}}$-TRS \cite{Ruy4, Ramos2018ExplicitPaths}, these properties justify the canonicity axiom from which contractibility at dimension $\geq 3$ is derived.
\end{rem}

\subsection{Higher Categories and Weak $\omega$-Groupoids}

We now recall the categorical structures that computational paths give rise to.

\begin{defi}[Globular Set]
A \emph{globular set} $X$ is a sequence of sets $X(0), X(1), X(2), \ldots$ equipped with source and target functions $\src, \tgt : X(n+1) \to X(n)$ satisfying the \emph{globular identities}:
\begin{align}
\src(\src(c)) &= \src(\tgt(c)) \label{eq:glob1}\\
\tgt(\src(c)) &= \tgt(\tgt(c)) \label{eq:glob2}
\end{align}
for all $c \in X(n+2)$.
\end{defi}

The globular identities ensure that higher cells have consistent boundaries. Intuitively, a 2-cell goes between two parallel 1-cells (1-cells with the same source and target), and a 3-cell goes between two parallel 2-cells, and so on.

Elements of $X(n)$ are called \emph{$n$-cells}. In our setting:
\begin{itemize}
    \item 0-cells are points (elements of type $A$)
    \item 1-cells are paths between points
    \item 2-cells are ``paths between paths'' (witnessed by $\RwEq$)
    \item $n$-cells for $n \geq 3$ are higher coherence witnesses
\end{itemize}

\begin{defi}[Weak $\omega$-Category {\cite{leinster1, batanin1}}]
A \emph{weak $\omega$-category} consists of:
\begin{enumerate}
    \item A globular set $\Cell_n$ for $n \in \mathbb{N}$
    \item An identity operation $\id : \Cell_n \to \Cell_{n+1}$ for each $n$
    \item A composition operation $\comp : \Cell_{n+1} \times_{\Cell_n} \Cell_{n+1} \to \Cell_{n+1}$ for composable cells
    \item Coherence witnesses (associators, unitors) at each level
    \item Higher coherences (pentagon, triangle, etc.) relating the coherence witnesses
    \item Coherences at all higher levels, satisfying a contractibility condition
\end{enumerate}
The ``weak'' qualifier means that categorical laws (associativity, unit laws) hold only up to coherent isomorphism, not strict equality.
\end{defi}

\begin{defi}[Weak $\omega$-Groupoid]
A \emph{weak $\omega$-groupoid} is a weak $\omega$-category in which every cell is invertible: for each $(n+1)$-cell $f$, there exists an $(n+1)$-cell $\inv(f)$ with $\src(\inv(f)) = \tgt(f)$, $\tgt(\inv(f)) = \src(f)$, together with $(n+2)$-cells witnessing $\inv(f) \comp f \sim \id$ and $f \comp \inv(f) \sim \id$.
\end{defi}

The fundamental insight connecting type theory to higher category theory is:

\begin{thm}[Lumsdaine \cite{Lumsdaine2009}, van den Berg--Garner \cite{BergGarner2011}]
For any type $A$ in Martin-L\"of type theory, the identity types on $A$ form a weak $\omega$-groupoid.
\end{thm}

The goal of this paper is to prove an analogous result for computational paths: that the explicit rewrite structure provides a weak $\omega$-groupoid with computable coherence witnesses.

\section{Cell Structure}
\label{sec:cells}

We now define the cells at each dimension for the weak $\omega$-groupoid of computational paths.

\subsection{Dimension 0: Points}

\begin{defi}[0-cells]
For a type $A$, the 0-cells are simply the elements of $A$:
\[
\Cell_0(A) := A
\]
\end{defi}

\subsection{Dimension 1: Paths}

\begin{defi}[1-cells]
A 1-cell is a computational path bundled with its endpoints:
\[
\Cell_1(A) := \{(a, b, p) \mid a, b : A, \, p : \Path(a, b)\}
\]
with source and target:
\begin{align*}
\src : \Cell_1(A) &\to \Cell_0(A), \quad \src(a, b, p) = a \\
\tgt : \Cell_1(A) &\to \Cell_0(A), \quad \tgt(a, b, p) = b
\end{align*}
\end{defi}

\subsection{Dimension 2: Derivations Between Paths}

\begin{defi}[2-cells]
\label{def:2cell}
A 2-cell is a \emph{derivation} between two paths with the same endpoints. Derivations are inductively defined:
\[
\Cell_2(p, q) \quad \text{for } p, q : \Path(a,b)
\]
with constructors:
\begin{itemize}
    \item $\refl(p) : \Cell_2(p, p)$ --- identity derivation
    \item $\mathrm{step}(s) : \Cell_2(p, q)$ where $s : \mathrm{Step}(p, q)$ --- single rewrite step
    \item $\inv(d) : \Cell_2(q, p)$ where $d : \Cell_2(p, q)$ --- inverse
    \item $d_1 \comp d_2 : \Cell_2(p, r)$ where $d_1 : \Cell_2(p, q)$ and $d_2 : \Cell_2(q, r)$ --- vertical composition
\end{itemize}

Source and target maps return the boundary paths. Given the implicit endpoints $a, b$ with $p, q : \Path(a,b)$:
\begin{align*}
\src(d) &:= (a, b, p) : \Cell_1(A) \\
\tgt(d) &:= (a, b, q) : \Cell_1(A)
\end{align*}
Thus $\src$ and $\tgt$ embed the path back into the bundled 1-cell type. When working fibrewise (with fixed $a, b$), we simply have $\src(d) = p$ and $\tgt(d) = q$.
\end{defi}

\begin{rem}[Key design: Derivations as data]
Crucially, 2-cells are \emph{data}, not mere propositions. This means different derivations between the same paths are distinguishable. The groupoid laws (associativity, units, inverses) hold \emph{up to higher cells}, not as strict equalities.

This is essential for a genuine weak $\omega$-groupoid: the coherence witnesses are non-trivial 3-cells.
\end{rem}

\begin{lem}[Equivalence of $\RwEq$ and $\Cell_2$ Inhabitedness]
\label{lem:rweq-cell2}
For paths $p, q : \Path(a, b)$, the following are equivalent:
\begin{enumerate}
    \item $p \RwEq q$ (rewrite equivalence holds)
    \item $\Cell_2(p, q)$ is inhabited (there exists a derivation from $p$ to $q$)
\end{enumerate}
\end{lem}

\begin{proof}
$(1 \Rightarrow 2)$: By induction on the derivation of $p \RwEq q$:
\begin{itemize}
    \item If $p = q$ (reflexivity), then $\refl(p) : \Cell_2(p, p)$.
    \item If $p \Rw q$ via a step $s$, then $\mathrm{step}(s) : \Cell_2(p, q)$.
    \item If $q \RwEq p$ with witness $d : \Cell_2(q, p)$, then $\inv(d) : \Cell_2(p, q)$.
    \item If $p \RwEq r$ and $r \RwEq q$ with witnesses $d_1, d_2$, then $d_1 \comp d_2 : \Cell_2(p, q)$.
\end{itemize}

$(2 \Rightarrow 1)$: By induction on the derivation $d : \Cell_2(p, q)$:
\begin{itemize}
    \item $\refl(p)$ gives $p \RwEq p$ by reflexivity.
    \item $\mathrm{step}(s)$ gives $p \RwEq q$ since $s : \mathrm{Step}(p, q)$ implies $p \Rw q$.
    \item $\inv(d)$ gives $q \RwEq p$ by symmetry (IH), hence $p \RwEq q$ by symmetry again.
    \item $d_1 \comp d_2$ gives $p \RwEq q$ by transitivity (IH on both).
\end{itemize}
\end{proof}

\begin{rem}[Clarification on Normal Forms]
An important consequence of Lemma~\ref{lem:rweq-cell2} combined with Assumption~\ref{assum:parallel-nf} is: if $\Cell_2(p, q)$ is inhabited (i.e., there exists a derivation $d : \Cell_2(p, q)$), then $\|p\| = \|q\|$. This follows because $\Cell_2(p, q)$ inhabited implies $p \RwEq q$ by Lemma~\ref{lem:rweq-cell2}, and for parallel paths the normal forms coincide by Assumption~\ref{assum:parallel-nf}.

Note that the converse also holds: any two parallel paths $p, q : \Path(a, b)$ have $\|p\| = \|q\|$ (by Assumption~\ref{assum:parallel-nf}), and hence the canonical derivation $\gamma_{p,q} : \Cell_2(p, q)$ exists. Thus $\Cell_2(p, q)$ is \emph{always} inhabited for parallel paths. The distinction is that $\RwEq$ is a proposition (proof-irrelevant), while $\Cell_2$ is data (different derivations are distinguishable).
\end{rem}

\subsection{Dimension 3: Meta-Derivations}

\begin{defi}[3-cells]
A 3-cell connects two parallel 2-cells (derivations with the same source and target paths):
\[
\Cell_3(d_1, d_2) \quad \text{for } d_1, d_2 : \Cell_2(p, q)
\]
with constructors:
\begin{itemize}
    \item $\refl(d) : \Cell_3(d, d)$
    \item $\mathrm{step}(s) : \Cell_3(d_1, d_2)$ where $s : \mathrm{MetaStep}_3(d_1, d_2)$
    \item $\inv(\alpha) : \Cell_3(d_2, d_1)$ where $\alpha : \Cell_3(d_1, d_2)$
    \item $\alpha \comp \beta : \Cell_3(d_1, d_3)$ where $\alpha : \Cell_3(d_1, d_2)$ and $\beta : \Cell_3(d_2, d_3)$
\end{itemize}
\end{defi}

\begin{defi}[Primitive Meta-Steps: $\mathrm{MetaStep}_3$]
\label{def:metastep3}
For 2-cells $d_1, d_2 : \Cell_2(p, q)$, the type $\mathrm{MetaStep}_3(d_1, d_2)$ is an inductive type with the following constructors:

\textbf{Groupoid laws} (for derivations $d, d', d'' : \Cell_2$):
\begin{align*}
\mathsf{vcomp\_refl\_right}(d) &: \mathrm{MetaStep}_3(d \comp \refl, d) \\
\mathsf{vcomp\_refl\_left}(d) &: \mathrm{MetaStep}_3(\refl \comp d, d) \\
\mathsf{vcomp\_assoc}(d, d', d'') &: \mathrm{MetaStep}_3((d \comp d') \comp d'', d \comp (d' \comp d'')) \\
\mathsf{inv\_inv}(d) &: \mathrm{MetaStep}_3(\inv(\inv(d)), d) \\
\mathsf{vcomp\_inv\_right}(d) &: \mathrm{MetaStep}_3(d \comp \inv(d), \refl) \\
\mathsf{vcomp\_inv\_left}(d) &: \mathrm{MetaStep}_3(\inv(d) \comp d, \refl)
\end{align*}

\textbf{Step coherence} (since $\mathrm{Step}$ is propositional):
\[
\mathsf{step\_eq}(s_1, s_2) : \mathrm{MetaStep}_3(\mathrm{step}(s_1), \mathrm{step}(s_2)) \quad \text{for } s_1, s_2 : \mathrm{Step}(p, q)
\]

\textbf{Canonicity axiom} (every derivation connects to the canonical derivation; see Definition~\ref{def:canonicity}):
\[
\mathsf{can}_d : \mathrm{MetaStep}_3(d, \gamma_{p,q}) \quad \text{for } d : \Cell_2(p, q)
\]

\textbf{Higher coherences}:
\begin{align*}
\mathsf{pentagon}(f, g, h, k) &: \mathrm{MetaStep}_3(\text{pentagonLeft}, \text{pentagonRight}) \\
\mathsf{triangle}(f, g) &: \mathrm{MetaStep}_3(\text{triangleLeft}, \text{triangleRight}) \\
\mathsf{interchange}(\alpha, \beta) &: \mathrm{MetaStep}_3(\text{interchangeLeft}, \text{interchangeRight})
\end{align*}

\textbf{Whiskering} (see Definition~\ref{def:whiskering}):
\begin{align*}
\mathsf{whiskerLeft}(h, d) &: \mathrm{MetaStep}_3(\ldots) \\
\mathsf{whiskerRight}(d, g) &: \mathrm{MetaStep}_3(\ldots)
\end{align*}
\end{defi}

\subsection{Dimension 4 and Higher}

\begin{defi}[Primitive Meta-Steps: $\mathrm{MetaStep}_n$ for $n \geq 4$]
\label{def:metastepn}
For $n \geq 4$ and $(n-1)$-cells $c_1, c_2 : \Cell_{n-1}(d_1, d_2)$, the type $\mathrm{MetaStep}_n(c_1, c_2)$ has constructors:
\begin{itemize}
    \item Groupoid laws: unit laws, associativity, inverse laws
    \item Step coherence: $\mathsf{step\_eq}$ for equal steps
    \item Canonicity axiom: $\mathsf{can}_n(c) : \mathrm{MetaStep}_n(c, \gamma_n)$ where $\gamma_n$ is the canonical $(n-1)$-cell (defined analogously to $\gamma_{p,q}$ at level 2)
    \item Whiskering operations
\end{itemize}
The higher coherences (pentagon, triangle, interchange) are only needed at level 3. Contractibility at dimension $\geq 4$ is derived from the canonicity axiom exactly as at level 3 (Theorem~\ref{thm:contractibility}).
\end{defi}

\begin{defi}[4-cells]
\[
\Cell_4(m_1, m_2) \quad \text{for } m_1, m_2 : \Cell_3(d_1, d_2)
\]
with constructors: $\refl$, $\mathrm{step}$ (from $\mathrm{MetaStep}_4$), $\inv$, and $\comp$.
\end{defi}

\begin{defi}[$n$-cells for $n \geq 5$]
For $n \geq 5$, $n$-cells connect parallel $(n-1)$-cells:
\[
\Cell_n(c_1, c_2) \quad \text{for } c_1, c_2 : \Cell_{n-1}(\ldots)
\]
with constructors: $\refl$, $\mathrm{step}$ (from $\mathrm{MetaStep}_n$), $\inv$, and $\comp$. This provides a uniform tower to arbitrary dimensions.
\end{defi}

\subsection{The Canonicity Axiom}

We now introduce the key axiom that grounds contractibility in the normalization algorithm. First, we establish notation for the normal form construction.

\begin{defi}[Normal Form]
\label{def:normal-form}
Every path $p : \Path(a, b)$ has a \emph{normal form} $\|p\|$, obtained by the normalization algorithm of the $\mathrm{LND}_{\mathrm{EQ}}$-TRS \cite{Ruy4, Ramos2018ExplicitPaths}. Concretely, $\|p\|$ is the canonical representative of the equivalence class $[p]$ under the rewrite relation $\RwEq$.
\end{defi}

\begin{lem}[Normal Forms and Rewrite Equivalence]
\label{lem:nf-equiv}
If $p, q : \Path(a, b)$ are rewrite-equivalent (i.e., there exists $d : \Cell_2(p, q)$), then $\|p\| = \|q\|$.
\end{lem}

\begin{proof}
By Definition~\ref{def:normal-form}, the normal form $\|p\|$ is the canonical representative of the equivalence class $[p]$ under the rewrite relation $\RwEq$. If $p \RwEq q$, then $p$ and $q$ belong to the same equivalence class, so they share the same canonical representative: $\|p\| = \|q\|$.

Concretely, since $\mathrm{LND}_{\mathrm{EQ}}$-TRS is terminating and confluent \cite{Ruy4, Ramos2018ExplicitPaths}, both $p$ and $q$ reduce to a common normal form.
\end{proof}

\begin{rem}[What This Lemma Does Not Claim]
\label{rem:not-all-parallel}
Note that Lemma~\ref{lem:nf-equiv} does \emph{not} assert that all parallel paths (paths with the same endpoints) share a normal form. That would be a much stronger claim, equivalent to saying that every hom-set $\Path(a,b)$ consists of a single $\RwEq$-equivalence class.

Instead, the lemma only applies to paths that are \emph{already known to be rewrite-equivalent}. This distinction is crucial: 2-cells witness rewrite equivalence, not mere parallelism. The weak $\omega$-groupoid structure is non-trivial precisely because different paths between the same endpoints may belong to different equivalence classes, with 2-cells only connecting paths within the same class.
\end{rem}

\begin{defi}[Normalizing Derivation]
\label{def:normalizing-deriv}
For each path $p$, the normalization algorithm produces a finite rewrite chain
\[
p = r_0 \Rw r_1 \Rw \cdots \Rw r_k = \|p\|
\]
We define the \emph{normalizing derivation} $\delta_p : \Cell_2(p, \|p\|)$ by composing the corresponding 2-cells:
\[
\delta_p := \mathrm{step}(s_1) \comp \mathrm{step}(s_2) \comp \cdots \comp \mathrm{step}(s_k)
\]
where each $s_i : \mathrm{Step}(r_{i-1}, r_i)$ is a single rewrite step in the chain. (If $p$ is already in normal form, $\delta_p = \refl_p$.)
\end{defi}

\begin{defi}[Canonical Derivation]
\label{def:canonical}
For paths $p, q : \Path(a, b)$ that are rewrite-equivalent (i.e., such that $\Cell_2(p, q)$ is inhabited), we define the \emph{canonical derivation} $\gamma_{p,q} : \Cell_2(p, q)$ as the following composition:
\[
\gamma_{p,q} \;:=\; \delta_p \comp \inv(\delta_q)
\]
where $\delta_p : \Cell_2(p, \|p\|)$ is the normalizing derivation from Definition~\ref{def:normalizing-deriv}.

By Lemma~\ref{lem:nf-equiv}, since $p \RwEq q$, we have $\|p\| = \|q\|$. Thus $\delta_p : \Cell_2(p, \|p\|)$ and $\inv(\delta_q) : \Cell_2(\|q\|, q) = \Cell_2(\|p\|, q)$, making the composition well-typed:
\[
\begin{tikzcd}
p \arrow[r, "\delta_p"] & \|p\| = \|q\| \arrow[r, "\inv(\delta_q)"] & q
\end{tikzcd}
\]

The canonical derivation $\gamma_{p,q}$ is only defined when $p$ and $q$ are already known to be rewrite-equivalent. It provides a \emph{distinguished} 2-cell among all 2-cells connecting $p$ to $q$---the one that factors through the shared normal form.
\end{defi}

\begin{defi}[Canonicity Axiom]
\label{def:canonicity}
The \emph{canonicity axiom} states that every derivation is connected to the canonical derivation by a 3-cell:
\[
\mathsf{can}_d : \Cell_3(d, \gamma_{p,q}) \quad \text{for each } d : \Cell_2(p, q)
\]
Equivalently, $\mathsf{can}$ is a constructor of $\mathrm{MetaStep}_3$:
\[
\mathsf{can} : \prod_{d : \Cell_2(p,q)} \mathrm{MetaStep}_3(d, \gamma_{p,q})
\]
At higher levels ($n \geq 4$), analogous canonicity axioms $\mathsf{can}_n$ are defined, with contractibility derived in the same way.

\textbf{Axiomatization:} In this formalization, we \emph{axiomatize} the existence of the 3-cell $\mathsf{can}_d$. This is the key primitive that grounds contractibility at dimension 3 and above. The axiom is not derived from more primitive principles within our formal system, but is justified semantically by the meta-theoretical properties described in Remark~\ref{rem:grounded-axiom}.
\end{defi}

\begin{rem}[Computational Content]
\label{rem:grounded-axiom}
The canonicity axiom is grounded in the concrete normalization algorithm:
\begin{enumerate}
    \item The normalizing derivation $\delta_p$ exists for every path (normalization terminates) \cite{Ruy4, Ramos2018ExplicitPaths}
    \item For rewrite-equivalent paths $p \RwEq q$, the canonical derivation $\gamma_{p,q}$ is uniquely determined by factoring through the shared normal form
    \item Confluence of the $\mathrm{LND}_{\mathrm{EQ}}$-TRS ensures that rewrite-equivalent paths share a normal form (Lemma~\ref{lem:nf-equiv})
\end{enumerate}
Unlike a bare contractibility axiom, the canonicity axiom connects each derivation to a \emph{specific, computable} target. Importantly, it only quantifies over \emph{existing} derivations $d : \Cell_2(p, q)$, not over all parallel paths.
\end{rem}

\begin{rem}[Comparison with HoTT]
\label{rem:hott-comparison}
The canonicity axiom plays a role analogous to J-elimination (path induction) in HoTT:
\begin{center}
\begin{tabular}{l|l}
\textbf{HoTT} & \textbf{Computational Paths} \\
\hline
J eliminator & Canonicity axiom \\
Path induction & Reduction to canonical derivation \\
$\refl_a$ is the canonical loop & $\|p\|$ is the canonical representative \\
\end{tabular}
\end{center}
Both encode the fundamental fact that identity proofs are unique up to higher identification---but the canonicity axiom provides explicit computational content via the normalization algorithm.
\end{rem}

\subsection{Contractibility from the Canonicity Axiom}

\begin{thm}[Contractibility at dimension $\geq 3$]
\label{thm:contractibility}
For any $n \geq 2$ and any two parallel $n$-cells $c_1, c_2$ (with $\src(c_1) = \src(c_2)$ and $\tgt(c_1) = \tgt(c_2)$), there exists an $(n+1)$-cell:
\[
\chi_n(c_1, c_2) : \Cell_{n+1}(c_1, c_2)
\]
In particular, any two parallel 2-cells (derivations) are connected by a 3-cell, and so on for higher dimensions.
\end{thm}

\begin{proof}
At level 3, contractibility is \emph{derived} from the canonicity axiom. Given parallel derivations $d_1, d_2 : \Cell_2(p, q)$, we construct:
\[
\chi_3(d_1, d_2) \;:=\; \mathsf{can}_{d_1} \comp \inv(\mathsf{can}_{d_2}) \;:\; \Cell_3(d_1, d_2)
\]
Both $d_1$ and $d_2$ are connected to the same canonical derivation $\gamma_{p,q}$, so composing $\mathsf{can}_{d_1} : \Cell_3(d_1, \gamma_{p,q})$ with $\inv(\mathsf{can}_{d_2}) : \Cell_3(\gamma_{p,q}, d_2)$ yields a 3-cell from $d_1$ to $d_2$:
\[
\begin{tikzcd}
d_1 \arrow[r, "\mathsf{can}_{d_1}"] & \gamma_{p,q} \arrow[r, "\inv(\mathsf{can}_{d_2})"] & d_2
\end{tikzcd}
\]

At higher levels ($n \geq 4$), analogous canonicity axioms $\mathsf{can}_n$ are defined, with contractibility derived in the same way.
\end{proof}

\begin{rem}[Loop contraction as a special case]
\label{rem:loop-contract-derived}
Loop contraction follows from the canonicity axiom and the groupoid laws. For a loop $d : \Cell_2(p, p)$, we have $\gamma_{p,p} = \delta_p \comp \inv(\delta_p) \sim \refl_p$ by the inverse law. Thus:
\[
\chi_3(d, \refl_p) \;:\; \Cell_3(d, \refl_p)
\]
shows that every loop on $p$ is connected to the reflexivity derivation.
\end{rem}

\begin{rem}[Batanin-style Contractibility]
\label{rem:batanin-contractibility}
The $\omega$-groupoid is \emph{contractible at dimension $k$} if any two parallel $(k-1)$-cells are connected by a $k$-cell. This is the \emph{Batanin-style} contractibility condition for higher coherence structures---it means that higher hom-spaces are contractible.

\textbf{Important:} This is \emph{not} the same as global homotopy contractibility (being equivalent to a point). Rather, it says that at sufficiently high dimensions (here, $\geq 3$), all parallel cells are connected. The underlying 0-cells (points of $A$) and 1-cells (paths) may have non-trivial structure; only the higher coherence data becomes trivial.

The canonicity axiom at level 3 directly yields contractibility at all dimensions $\geq 3$.
\end{rem}

\section{Operations}
\label{sec:operations}

\subsection{Identity}

For each dimension, we define identity cells.

\begin{defi}[Identity at dimension 1]
For $a : \Cell_0(A)$:
\[
\id_0(a) := (a, a, \refl_a) : \Cell_1(A)
\]
\end{defi}

\begin{defi}[Identity at dimension 2]
For a 1-cell $c = (a, b, p) : \Cell_1(A)$, the identity 2-cell is the reflexivity derivation:
\[
\id_1(c) := \refl(p) : \Cell_2(p, p)
\]
This is the trivial derivation witnessing that $p$ is equivalent to itself.
\end{defi}

\begin{defi}[Identity at dimension $n \geq 3$]
For an $(n-1)$-cell $c : \Cell_{n-1}$, the identity $n$-cell is:
\[
\id_{n-1}(c) := \refl(c) : \Cell_n(c, c)
\]
At each level, this is the reflexivity constructor from the definition of $\Cell_n$.
\end{defi}

\begin{lem}[Identity source/target]
For all $n$ and $c : \Cell_n(A)$:
\begin{align*}
\src(\id_n(c)) &= c \\
\tgt(\id_n(c)) &= c
\end{align*}
\end{lem}

\begin{proof}
By definition of $\id_n$.
\end{proof}

\subsection{Composition}

Composition combines two cells that are ``composable'' (the target of the first equals the source of the second).

\begin{defi}[Composition at dimension 1]
For $f = (a, b, p)$ and $g = (b', c, q)$ with $\tgt(f) = \src(g)$ (i.e., $b = b'$):
\[
f \comp g := (a, c, \trans(p, q)) : \Cell_1(A)
\]
\end{defi}

\begin{defi}[Composition at dimension 2]
For derivations $d_1 : \Cell_2(p, q)$ and $d_2 : \Cell_2(q, r)$ with $\tgt(d_1) = \src(d_2)$, their \emph{vertical composition} is:
\[
d_1 \comp d_2 : \Cell_2(p, r)
\]
This is the $\comp$ constructor from Definition~\ref{def:2cell}, which composes two derivations ``end to end'' along a shared path.
\end{defi}

\begin{defi}[Composition at dimension $n \geq 3$]
For $n$-cells $c_1 : \Cell_n(d_1, d_2)$ and $c_2 : \Cell_n(d_2, d_3)$ with $\tgt(c_1) = \src(c_2)$, their composition is:
\[
c_1 \comp c_2 : \Cell_n(d_1, d_3)
\]
This is the $\comp$ constructor from the definition of $\Cell_n$. Composition at all levels follows the same pattern: composing two cells ``end to end'' along a shared boundary.
\end{defi}

\begin{lem}[Composition source/target]
For all $n$ and composable $f, g : \Cell_{n+1}(A)$:
\begin{align*}
\src(f \comp g) &= \src(f) \\
\tgt(f \comp g) &= \tgt(g)
\end{align*}
\end{lem}

\begin{proof}
By definition of composition at each level.
\end{proof}

\subsection{Inverse}

\begin{defi}[Inverse at dimension 1]
For $f = (a, b, p) : \Cell_1(A)$:
\[
\inv(f) := (b, a, \symm(p)) : \Cell_1(A)
\]
\end{defi}

\begin{defi}[Inverse at dimension 2]
For a 2-cell $d : \Cell_2(p, q)$, the inverse is:
\[
\inv(d) : \Cell_2(q, p)
\]
constructed using the $\inv$ constructor from Definition~\ref{def:2cell}.
\end{defi}

\begin{defi}[Inverse at dimension $n \geq 3$]
For an $n$-cell $c : \Cell_n(c_1, c_2)$, the inverse is:
\[
\inv(c) : \Cell_n(c_2, c_1)
\]
At each level, this uses the inverse constructor from the definition of $\Cell_n$.
\end{defi}

\begin{lem}[Inverse source/target]
For all $n$ and $f : \Cell_{n+1}(A)$:
\begin{align*}
\src(\inv(f)) &= \tgt(f) \\
\tgt(\inv(f)) &= \src(f)
\end{align*}
\end{lem}

\begin{proof}
By definition of inverse at each level.
\end{proof}

\section{Globular Identities}
\label{sec:globular}

The globular identities ensure that the source and target maps are compatible across dimensions.

\begin{thm}[Globular Identities]
\label{thm:globular}
For all $n \geq 0$ and $c : \Cell_{n+2}(A)$:
\begin{align}
\src(\src(c)) &= \src(\tgt(c)) \label{eq:glob-ss-st} \\
\tgt(\src(c)) &= \tgt(\tgt(c)) \label{eq:glob-ts-tt}
\end{align}
\end{thm}

\begin{proof}
We verify by cases on $n$.

\textbf{Case $n = 0$:} Here $c : \Cell_2(p, q)$ is a derivation between paths $p, q : \Path(a, b)$.
\begin{itemize}
    \item $\src(c) = p$ and $\tgt(c) = q$, where both $p$ and $q$ have source $a$ and target $b$.
    \item $\src(\src(c)) = \src(p) = a$ and $\src(\tgt(c)) = \src(q) = a$, so \eqref{eq:glob-ss-st} holds.
    \item $\tgt(\src(c)) = \tgt(p) = b$ and $\tgt(\tgt(c)) = \tgt(q) = b$, so \eqref{eq:glob-ts-tt} holds.
\end{itemize}

\textbf{Case $n = 1$:} Here $c : \Cell_3(d_1, d_2)$ where $d_1, d_2 : \Cell_2(p, q)$ are parallel derivations.
\begin{itemize}
    \item $\src(c) = d_1$ and $\tgt(c) = d_2$, both with source path $p$ and target path $q$.
    \item $\src(\src(c)) = \src(d_1) = p$ and $\src(\tgt(c)) = \src(d_2) = p$, so \eqref{eq:glob-ss-st} holds.
    \item $\tgt(\src(c)) = \tgt(d_1) = q$ and $\tgt(\tgt(c)) = \tgt(d_2) = q$, so \eqref{eq:glob-ts-tt} holds.
\end{itemize}

\textbf{Case $n \geq 2$:} Here $c : \Cell_{n+2}(c_1, c_2)$ where $c_1, c_2 : \Cell_{n+1}(d_1, d_2)$ are parallel $(n+1)$-cells.
\begin{itemize}
    \item $\src(c) = c_1$ and $\tgt(c) = c_2$, both with source $d_1$ and target $d_2$.
    \item $\src(\src(c)) = \src(c_1) = d_1$ and $\src(\tgt(c)) = \src(c_2) = d_1$, so \eqref{eq:glob-ss-st} holds.
    \item $\tgt(\src(c)) = \tgt(c_1) = d_2$ and $\tgt(\tgt(c)) = \tgt(c_2) = d_2$, so \eqref{eq:glob-ts-tt} holds.
\end{itemize}
In each case, the globular identities follow from the fact that parallel cells share their source and target.
\end{proof}

\section{Coherence Laws}
\label{sec:coherence}

We now prove that the groupoid operations satisfy the expected laws up to higher cells.

\subsection{Associativity}

\begin{thm}[Associativity Coherence]
\label{thm:assoc}
For composable 1-cells $f, g, h$, there exists a 2-cell:
\[
\assoc_{f,g,h} : (f \comp g) \comp h \longrightarrow f \comp (g \comp h)
\]
At higher dimensions, associativity is witnessed by cells whose boundaries are the two associated compositions.
\end{thm}

\begin{proof}
\textbf{Dimension 1:} Let $f = (a, b, p)$, $g = (b, c, q)$, $h = (c, d, r)$. Then:
\begin{itemize}
    \item $(f \comp g) \comp h = (a, d, \trans(\trans(p, q), r))$
    \item $f \comp (g \comp h) = (a, d, \trans(p, \trans(q, r)))$
\end{itemize}
Both have the same source $a$ and target $d$. The rewrite system includes the rule:
\[
\trans(\trans(p, q), r) \Rw \trans(p, \trans(q, r)) \quad (\text{associativity rule } \rhd_{tt})
\]
Thus, the two paths are $\RwEq$-equivalent via the rewrite rule $\rhd_{tt}$. The associator 2-cell is the derivation:
\[
\assoc_{f,g,h} := \mathrm{step}(\rhd_{tt}) : \Cell_2(\trans(\trans(p,q),r), \, \trans(p,\trans(q,r)))
\]
This witnesses that the two associated compositions are equivalent paths.

\textbf{Dimension 2:} For derivations $d_1, d_2, d_3 : \Cell_2$, we construct $\assoc_{d_1,d_2,d_3} : \Cell_3$ using the $\mathsf{vcomp\_assoc}$ constructor from $\mathrm{MetaStep}_3$ (Definition~\ref{def:metastep3}):
\[
\assoc_{d_1,d_2,d_3} := \mathrm{step}(\mathsf{vcomp\_assoc}(d_1, d_2, d_3)) : \Cell_3((d_1 \comp d_2) \comp d_3, \, d_1 \comp (d_2 \comp d_3))
\]

\textbf{Dimension $\geq 3$:} For $n$-cells $c_1, c_2, c_3$ at dimension $n \geq 3$, the associator $(n+1)$-cell is constructed using the $\mathsf{vcomp\_assoc}$ constructor from $\mathrm{MetaStep}_{n+1}$. By contractibility (Theorem~\ref{thm:contractibility}), all higher associators are connected, ensuring coherence at all levels.
\end{proof}

\subsection{Unit Laws}

\begin{thm}[Left Unit Coherence]
\label{thm:left-unit}
For any 1-cell $f$, there exists a 2-cell:
\[
\lunit_f : \id(\src(f)) \comp f \longrightarrow f
\]
\end{thm}

\begin{proof}
Let $f = (a, b, p)$. Then:
\begin{itemize}
    \item $\id(\src(f)) = \id(a) = (a, a, \refl_a)$
    \item $\id(\src(f)) \comp f = (a, b, \trans(\refl_a, p))$
\end{itemize}
The rewrite system includes the rule:
\[
\trans(\refl, p) \Rw p \quad (\text{left unit rule } \rhd_{tlr})
\]
Thus $\trans(\refl, p) \RwEq p$. The left unitor is the derivation:
\[
\lunit_f := \mathrm{step}(\rhd_{tlr}) : \Cell_2(\trans(\refl, p), \, p)
\]
This witnesses that composing with the identity path on the left yields an equivalent path.

\textbf{Higher dimensions:} At dimension 2, for a derivation $d : \Cell_2(p, q)$, the left unitor $\lunit_d : \Cell_3$ is constructed using the $\mathsf{vcomp\_refl\_left}$ constructor from $\mathrm{MetaStep}_3$:
\[
\lunit_d := \mathrm{step}(\mathsf{vcomp\_refl\_left}(d)) : \Cell_3(\refl \comp d, \, d)
\]
At dimension $\geq 3$, the left unitor is similarly constructed using the $\mathsf{vcomp\_refl\_left}$ constructor from $\mathrm{MetaStep}_n$.
\end{proof}

\begin{thm}[Right Unit Coherence]
\label{thm:right-unit}
For any 1-cell $f$, there exists a 2-cell:
\[
\runit_f : f \comp \id(\tgt(f)) \longrightarrow f
\]
\end{thm}

\begin{proof}
Let $f = (a, b, p)$. Then $f \comp \id(\tgt(f)) = (a, b, \trans(p, \refl_b))$. The rewrite rule:
\[
\trans(p, \refl) \Rw p \quad (\text{right unit rule } \rhd_{trr})
\]
gives $\trans(p, \refl) \RwEq p$. The right unitor is the derivation:
\[
\runit_f := \mathrm{step}(\rhd_{trr}) : \Cell_2(\trans(p, \refl), \, p)
\]
This witnesses that composing with the identity path on the right yields an equivalent path. At higher dimensions, the right unitor is constructed using the $\mathsf{vcomp\_refl\_right}$ constructor from $\mathrm{MetaStep}_n$.
\end{proof}

\subsection{Inverse Laws}

\begin{thm}[Left Inverse Coherence]
\label{thm:left-inv}
For any 1-cell $f$, there exists a 2-cell:
\[
\linv{f} : \inv(f) \comp f \longrightarrow \id(\tgt(f))
\]
\end{thm}

\begin{proof}
Let $f = (a, b, p)$. Then:
\begin{itemize}
    \item $\inv(f) = (b, a, \symm(p))$
    \item $\inv(f) \comp f = (b, b, \trans(\symm(p), p))$
    \item $\id(\tgt(f)) = \id(b) = (b, b, \refl_b)$
\end{itemize}
The rewrite rule:
\[
\trans(\symm(p), p) \Rw \refl \quad (\text{left inverse rule } \rhd_{tsr})
\]
gives $\trans(\symm(p), p) \RwEq \refl$. The left inverse witness is the derivation:
\[
\linv{f} := \mathrm{step}(\rhd_{tsr}) : \Cell_2(\trans(\symm(p), p), \, \refl)
\]
At higher dimensions, the left inverse witness is constructed using the $\mathsf{vcomp\_inv\_left}$ constructor from $\mathrm{MetaStep}_n$.
\end{proof}

\begin{thm}[Right Inverse Coherence]
\label{thm:right-inv}
For any 1-cell $f$, there exists a 2-cell:
\[
\rinv{f} : f \comp \inv(f) \longrightarrow \id(\src(f))
\]
\end{thm}

\begin{proof}
Let $f = (a, b, p)$. Then $f \comp \inv(f) = (a, a, \trans(p, \symm(p)))$ and $\id(\src(f)) = (a, a, \refl_a)$. The rewrite rule:
\[
\trans(p, \symm(p)) \Rw \refl \quad (\text{right inverse rule } \rhd_{tr})
\]
gives $\trans(p, \symm(p)) \RwEq \refl$. The right inverse witness is the derivation:
\[
\rinv{f} := \mathrm{step}(\rhd_{tr}) : \Cell_2(\trans(p, \symm(p)), \, \refl)
\]
At higher dimensions, the right inverse witness is constructed using the $\mathsf{vcomp\_inv\_right}$ constructor from $\mathrm{MetaStep}_n$.
\end{proof}

\subsection{Double Symmetry}

\begin{thm}[Double Symmetry]
For any 1-cell $f$, there exists a 2-cell:
\[
\inv(\inv(f)) \longrightarrow f
\]
\end{thm}

\begin{proof}
Using $\symm(\symm(p)) \Rw p$ (rule $\rhd_{ss}$).
\end{proof}

\section{Higher Coherences}
\label{sec:higher}

A weak $\omega$-groupoid requires not just the basic coherence laws, but also higher coherences that relate them. The key ones are the pentagon and triangle identities.

\subsection{Whiskering}

Before stating the pentagon and triangle coherences, we define the whiskering operations that appear in their formulations.

\begin{defi}[Horizontal Composition]
\label{def:hcomp}
Given 2-cells $\alpha : \Cell_2(f, f')$ where $f, f' : a \to b$, and $\beta : \Cell_2(g, g')$ where $g, g' : b \to c$, the \emph{horizontal composition} is:
\[
\alpha \star_h \beta : \Cell_2(f \comp g, f' \comp g')
\]
Horizontal composition combines 2-cells ``side by side'' (along a shared boundary point), in contrast to vertical composition $\comp$ which combines 2-cells ``end to end'' (along a shared 1-cell).

In a general 2-category, horizontal composition is a primitive operation. In our setting, we can define it via whiskering:
\[
\alpha \star_h \beta := \mathrm{whiskerRight}(\alpha, g) \comp \mathrm{whiskerLeft}(f', \beta)
\]
or equivalently:
\[
\alpha \star_h \beta := \mathrm{whiskerLeft}(f, \beta) \comp \mathrm{whiskerRight}(\alpha, g')
\]
The interchange law (Theorem~\ref{thm:interchange}) ensures these two definitions yield equivalent 2-cells.
\end{defi}

\begin{defi}[Whiskering]
\label{def:whiskering}
Given a 2-cell $\alpha : \Cell_2(f, f')$ where $f, f' : a \to b$, and 1-cells $g : b \to c$ and $h : z \to a$, we define:
\begin{itemize}
    \item \textbf{Right whiskering}: $\mathrm{whiskerRight}(\alpha, g) : \Cell_2(f \comp g, f' \comp g)$
    \item \textbf{Left whiskering}: $\mathrm{whiskerLeft}(h, \alpha) : \Cell_2(h \comp f, h \comp f')$
\end{itemize}
These operations express the functoriality of composition: composing with a fixed 1-cell preserves 2-cells.
\end{defi}

\begin{rem}[Whiskering as Horizontal Composition]
Whiskering is horizontal composition with an identity 2-cell:
\begin{align*}
\mathrm{whiskerRight}(\alpha, g) &= \alpha \star_h \id_g \\
\mathrm{whiskerLeft}(h, \alpha) &= \id_h \star_h \alpha
\end{align*}
In our construction, whiskering is a primitive operation in $\mathrm{MetaStep}_3$ (Definition~\ref{def:metastep3}), and horizontal composition is derived from it.
\end{rem}

\subsection{Pentagon Coherence}

The pentagon identity states that the two ways of reassociating four composable cells are connected by a coherent family of associators.

\begin{thm}[Pentagon Coherence]
\label{thm:pentagon}
For composable 1-cells $f, g, h, k$, there exists a 3-cell:
\[
\pi_{f,g,h,k} : \Cell_3(\mathrm{pentagonLeft}(f,g,h,k), \, \mathrm{pentagonRight}(f,g,h,k))
\]
where:
\begin{align*}
\mathrm{pentagonLeft} &= \assoc_{f \comp g, h, k} \comp \assoc_{f, g, h \comp k} \\
\mathrm{pentagonRight} &= (\mathrm{whiskerRight}(\assoc_{f,g,h}, k) \comp \assoc_{f, g \comp h, k}) \comp \mathrm{whiskerLeft}(f, \assoc_{g,h,k})
\end{align*}
\end{thm}

\begin{proof}
The pentagon 3-cell is constructed as a \emph{primitive} $\mathrm{MetaStep}_3$:
\[
\mathrm{pentagon}(f, g, h, k) : \mathrm{MetaStep}_3(\mathrm{pentagonLeft}, \mathrm{pentagonRight})
\]
This encodes the pentagon identity directly from the interaction of the associativity rewrite rules in the $\mathrm{LND}_{\mathrm{EQ}}$-TRS. Specifically, it states that the two composite associators around the pentagon are connected by a primitive 3-cell.
\end{proof}

\subsection{Triangle Coherence}

The triangle identity ensures compatibility between the associator and the unit laws.

\begin{thm}[Triangle Coherence]
\label{thm:triangle}
For composable 1-cells $f, g$, there exists a 3-cell:
\[
\theta_{f,g} : \Cell_3(\mathrm{triangleLeft}(f,g), \, \mathrm{triangleRight}(f,g))
\]
where:
\begin{align*}
\mathrm{triangleLeft} &= \assoc_{f, \id, g} \comp \mathrm{whiskerLeft}(f, \lunit_g) \\
\mathrm{triangleRight} &= \mathrm{whiskerRight}(\runit_f, g)
\end{align*}
\end{thm}

\begin{proof}
The triangle 3-cell is constructed as a primitive $\mathrm{MetaStep}_3$:
\[
\mathrm{triangle}(f, g) : \mathrm{MetaStep}_3(\mathrm{triangleLeft}, \mathrm{triangleRight})
\]
This encodes the triangle identity directly from the interaction of associativity and unit rewrite rules.
\end{proof}

\subsection{Interchange}

The interchange law ensures that horizontal and vertical composition commute appropriately.

\begin{thm}[Interchange]
\label{thm:interchange}
For 2-cells $\alpha : f \to f'$ and $\beta : g \to g'$ where $f, f' : a \to b$ and $g, g' : b \to c$, there exists a 3-cell:
\[
\mathrm{interchange}_{\alpha,\beta} : (\alpha \star_h \id_g) \comp (\id_{f'} \star_h \beta) \longrightarrow (\id_f \star_h \beta) \comp (\alpha \star_h \id_{g'})
\]
where $\star_h$ denotes horizontal composition (whiskering).
\end{thm}

\begin{proof}
Constructed as a primitive $\mathrm{MetaStep}_3.\mathrm{interchange}(\alpha, \beta)$.
\end{proof}

\subsection{Contractibility at Dimension $\geq 3$}

The defining property of a weak $\omega$-groupoid is that parallel cells at each dimension are connected by a cell one dimension higher.

As established in Theorem~\ref{thm:contractibility}, for any $n \geq 2$ and any two parallel $n$-cells $c_1, c_2$, there exists an $(n+1)$-cell $\chi_n(c_1, c_2) : \Cell_{n+1}(c_1, c_2)$. This is \emph{derived} from the canonicity axiom (Definition~\ref{def:canonicity}).

\begin{rem}[Non-truncation]
Unlike previous approaches that resulted in 2-truncated structures, our construction provides a \emph{full} weak $\omega$-groupoid with non-trivial cells at all dimensions. The key is that:
\begin{enumerate}
    \item $\Cell_n$ for $n \geq 2$ lives in the data universe, not propositions
    \item The canonicity axiom is a \emph{primitive} meta-step connecting derivations to a canonical form
    \item Contractibility is \emph{derived} from the canonicity axiom via the construction $\chi_n(c_1, c_2) := \mathsf{can}_{c_1} \comp \inv(\mathsf{can}_{c_2})$
\end{enumerate}
This matches the Lumsdaine/van den Berg-Garner construction where J-elimination provides contractibility.
\end{rem}

\section{Main Theorem}
\label{sec:main}

We now assemble the preceding results into our main theorem.

\begin{thm}[Computational Paths Form a Weak $\omega$-Groupoid]
\label{thm:main}
For any type $A$, the computational paths on $A$ form a weak $\omega$-groupoid with:
\begin{enumerate}
    \item Cell types $\Cell_n(A)$ at each dimension $n \geq 0$
    \item Source and target maps satisfying globular identities (Theorem~\ref{thm:globular})
    \item Identity operation with correct source/target
    \item Composition operation for composable cells with correct source/target
    \item Inverse operation with exchanged source/target
    \item Coherence witnesses:
    \begin{itemize}
        \item Associativity $\assoc_{f,g,h}$ (Theorem~\ref{thm:assoc})
        \item Left unit $\lunit_f$ (Theorem~\ref{thm:left-unit})
        \item Right unit $\runit_f$ (Theorem~\ref{thm:right-unit})
        \item Left inverse $\linv{f}$ (Theorem~\ref{thm:left-inv})
        \item Right inverse $\rinv{f}$ (Theorem~\ref{thm:right-inv})
    \end{itemize}
    \item Higher coherences:
    \begin{itemize}
        \item Pentagon coherence (Theorem~\ref{thm:pentagon})
        \item Triangle coherence (Theorem~\ref{thm:triangle})
    \end{itemize}
    \item Contractibility at dimension $\geq 3$ (Theorem~\ref{thm:contractibility})
\end{enumerate}
\end{thm}

\begin{proof}
Each component has been proven in the preceding sections. The key points are:

\textbf{Dimensions 0-1:} These contain genuine content---points of $A$ and paths between them. The groupoid operations are the standard ones from the theory of computational paths.

\textbf{Dimension 2:} 2-cells are \emph{derivations} between parallel paths, built inductively from $\mathrm{Step}$, inverse, and composition. The coherence laws (associativity, units, inverses) hold up to 3-cells because the $\mathrm{LND}_{\mathrm{EQ}}$-TRS rewrite rules provide the necessary witnesses. For example, $\trans(\trans(p,q),r) \RwEq \trans(p,\trans(q,r))$ gives the associator as a 2-cell.

\textbf{Dimension 3:} 3-cells connect parallel 2-cells (derivations with the same source and target paths). The pentagon and triangle coherences are primitive constructors of $\mathrm{MetaStep}_3$. The canonicity axiom provides $\mathsf{can}_d : \Cell_3(d, \gamma_{p,q})$ for each derivation $d$, from which contractibility is derived.

\textbf{Dimension $\geq 4$:} Higher cells continue the pattern: $n$-cells connect parallel $(n-1)$-cells. Contractibility is derived from the canonicity axiom: given any two parallel cells $c_1, c_2 : \Cell_n(x, y)$ for $n \geq 2$, we construct $\chi_n(c_1, c_2) := \mathsf{can}_{c_1} \comp \inv(\mathsf{can}_{c_2}) : \Cell_{n+1}(c_1, c_2)$. This ensures all higher coherences are satisfied.

The structure is a \emph{full} weak $\omega$-groupoid following the Lumsdaine/van den Berg-Garner definition, with contractibility at dimension $\geq 3$ derived from the canonicity axiom.
\end{proof}

\subsection{Comparison with Traditional Results}
\label{sec:comparison}

\begin{thm}[Comparison with Identity Type $\omega$-Groupoid]
\label{thm:comparison}
The weak $\omega$-groupoid structure on computational paths corresponds to the identity type $\omega$-groupoid from \cite{Lumsdaine2009, BergGarner2011} in the following sense:

\begin{enumerate}
    \item At dimensions 0 and 1, the structures coincide exactly.
    \item At dimension 2, there is a bijective correspondence: a 2-cell in our structure (witnessing $p \RwEq q$) corresponds to an inhabitant of $\Id_{\Path(a,b)}(p, q)$ in the identity type formulation.
    \item At dimension $\geq 3$, both structures exhibit contractibility for the class of types we consider.
\end{enumerate}
\end{thm}

\begin{proof}[Sketch]
The computational paths identity type and the traditional identity type are propositionally equivalent at dimension 1: a path $a =_s b$ gives rise to $s(a,b) : \Id_A(a,b)$, and conversely, any term of the identity type can be represented as a computational path via the $\mathrm{REWR}$ eliminator.

At dimension 2, the correspondence $\RwEq \leftrightarrow \Id_{\Path}$ holds because both capture ``sameness of paths.''

At dimension 3 and above, contractibility is derived from the canonicity axiom, which is analogous to the J-elimination principle in homotopy type theory. The canonicity axiom states that every derivation connects to a canonical derivation (through normal forms), from which contractibility follows: $\chi_n(c_1, c_2) := \mathsf{can}_{c_1} \comp \inv(\mathsf{can}_{c_2})$.
\end{proof}

\begin{rem}[Semantic Justification of the Canonicity Axiom]
The canonicity axiom is justified semantically by the normalization and confluence properties of the $\mathrm{LND}_{\mathrm{EQ}}$-TRS rewrite system \cite{Ruy4, Ramos2018ExplicitPaths}, combined with the fact that rewrite steps live in a proof-irrelevant universe. Because Step is proof-irrelevant and the system is normalizing and confluent, all derivations between the same endpoints can be connected to the same canonical derivation $\gamma_{p,q}$ through normal forms.
\end{rem}

\begin{rem}[Role of the Canonicity Axiom]
The canonicity axiom plays a role analogous to the J-elimination rule in Martin-L\"of type theory. Just as J cannot be derived from pure type theory but must be postulated (justified by the computation rule), the canonicity axiom cannot be derived from pure groupoid algebra but is justified by the computational properties of the underlying rewrite system---specifically, the existence of normal forms and the canonical derivation $\gamma_{p,q} = \delta_p \comp \inv(\delta_q)$. This axiom is the key ingredient that elevates our structure from a weak 2-groupoid to a full weak $\omega$-groupoid. Importantly, unlike a bare contractibility axiom, the canonicity axiom has a concrete, canonical target grounded in the normalization algorithm.
\end{rem}

\section{Conclusion}
\label{sec:conclusion}

\subsection{Summary}

We have provided a complete pen-and-paper proof that computational paths form a \emph{full} weak $\omega$-groupoid in the sense of Lumsdaine and van den Berg-Garner. This result:

\begin{enumerate}
    \item \textbf{Correct 2-cell definition:} 2-cells are \emph{explicitly witnessed} by $\RwEq$ proofs, ensuring that 2-cells only exist between rewrite-equivalent paths. This avoids the ``indiscrete'' structure that would result from merely requiring parallelism.

    \item \textbf{Genuine coherences from LNDEQ-TRS:} The associator uses $\mathsf{trans\_assoc}$, unitors use $\mathsf{trans\_refl\_*}$, and inverse laws use $\mathsf{trans\_symm}$/$\mathsf{symm\_trans}$. These are not vacuous---they genuinely construct the coherence witnesses from the rewrite rules.

    \item \textbf{Full tower structure:} We define cells at all dimensions: $\Cell_2$ (derivations between paths), $\Cell_3$ (3-cells between derivations), $\Cell_4$ (4-cells), and $\Cell_n$ for arbitrary $n$. This gives a proper weak $\omega$-groupoid, not merely a 2-groupoid.

    \item \textbf{Canonicity axiom:} The key ingredient is the canonicity axiom, which states that every derivation connects to a canonical derivation $\gamma_{p,q}$ through normal forms. This axiom, analogous to J-elimination in HoTT, is justified semantically by normalization and confluence of the underlying rewrite system \cite{Ruy4, Ramos2018ExplicitPaths}. Unlike a bare contractibility axiom, the canonicity axiom has a concrete, canonical target.

    \item \textbf{Derived contractibility:} At dimension $\geq 3$, contractibility is \emph{derived} from the canonicity axiom via a simple construction: $\chi_n(c_1, c_2) := \mathsf{can}_{c_1} \comp \inv(\mathsf{can}_{c_2})$. This is both elegant and principled---contractibility follows from the existence of canonical forms rather than being postulated directly.
\end{enumerate}

\subsection{Relation to the Formalization}

The constructions in this paper have been formalized in Lean 4 and are available at \cite{ComputationalPathsLean}. The formalization implements:

\begin{itemize}
    \item The full cell tower ($\Cell_2$, $\Cell_3$, $\Cell_4$, and $\Cell_n$ for higher $n$)
    \item All groupoid operations at each level with correct source/target behavior
    \item The canonicity axiom connecting derivations to canonical forms
    \item Derived contractibility functions $\chi_n$ at each level
    \item Pentagon and triangle coherences as primitive constructors
    \item The interchange law and all required higher coherences
\end{itemize}

The formalized code closely mirrors the pen-and-paper proofs presented here, providing machine-checked verification of the weak $\omega$-groupoid structure. The canonicity axiom yields particularly elegant code: contractibility at each level is derived via a simple composition with the canonical derivation.

\subsection{Future Work}

Several directions for future research emerge:

\begin{enumerate}
    \item \textbf{Non-contractible higher structure:} Our weak $\omega$-groupoid has contractible higher cells (dimension $\geq 3$). To capture types with genuinely non-trivial higher homotopy (e.g., $S^1$ with $\pi_1 \cong \mathbb{Z}$), one would need:
    \begin{itemize}
        \item Placing $\RwEq$ in Type rather than Prop, making rewrite derivations distinguishable data
        \item Developing a meta-level rewrite system for comparing rewrite derivations
        \item Extending the Derivation tower with non-trivial higher cells
    \end{itemize}

    \item \textbf{Univalence:} Investigate computational versions of the univalence axiom in the context of computational paths.

    \item \textbf{Higher Inductive Types:} Extend the framework to handle HITs with explicit path constructors.

    \item \textbf{Decidability:} Study decidability properties of $\RwEq$ and implications for type checking.

    \item \textbf{Applications:} Apply explicit coherence witnesses to problems in algebra and topology.
\end{enumerate}


\section*{Acknowledgment}

The authors thank the anonymous reviewers for their helpful comments and suggestions.


\bibliographystyle{alphaurl}
\bibliography{references}

@article{Queiroz2016Paths,
  title = {Propositional equality, identity types, and direct computational paths},
  author = {de Queiroz, Ruy J. G. B. and de Oliveira, Anjolina G. and Ramos, Arthur F.},
  journal = {South American Journal of Logic},
  volume = {2},
  number = {2},
  pages = {245--296},
  year = {2016},
  note = {Special Issue: A Festschrift for Francisco Miraglia}
}

@article{Ramos2018ExplicitPaths,
  title = {Explicit Computational Paths},
  author = {Ramos, Arthur F. and de Queiroz, Ruy J. G. B. and de Oliveira, Anjolina G. and de Veras, Tiago M. L.},
  journal = {South American Journal of Logic},
  volume = {4},
  number = {2},
  pages = {441--484},
  year = {2018},
  issn = {2446-6719},
  url = {https://www.sa-logic.org/sajl-v4-i2/10-Ramos-de%20Queiroz-de%20Oliveira-de-Veras-SAJL.pdf}
}

@article{Ramos2017IdentityPaths,
  title = {On the Identity Type as the Type of Computational Paths},
  author = {Ramos, Arthur F. and de Queiroz, Ruy J. G. B. and de Oliveira, Anjolina G.},
  journal = {Logic Journal of the IGPL},
  volume = {25},
  number = {4},
  pages = {562--584},
  year = {2017},
  publisher = {Oxford University Press},
  doi = {10.1093/jigpal/jzx015},
  url = {https://ieeexplore.ieee.org/document/8204959}
}

@article{Veras2023WeakGroupoid,
  title = {Computational Paths -- a Weak Groupoid},
  author = {de Veras, Tiago M. L. and Ramos, Arthur F. and de Queiroz, Ruy J. G. B. and de Oliveira, Anjolina G.},
  journal = {Journal of Logic and Computation},
  volume = {35},
  number = {5},
  pages = {exad071},
  year = {2023},
  doi = {10.1093/logcom/exad071},
  url = {https://doi.org/10.1093/logcom/exad071}
}

@inproceedings{HofmannStreicher1994,
  title = {The groupoid model refutes uniqueness of identity proofs},
  author = {Hofmann, Martin and Streicher, Thomas},
  booktitle = {Proceedings of the Ninth Annual IEEE Symposium on Logic in Computer Science (LICS)},
  pages = {208--212},
  year = {1994},
  publisher = {IEEE},
  address = {Piscataway, NJ, USA}
}

@article{BergGarner2011,
  title = {Types are weak $\omega$-groupoids},
  author = {van den Berg, Benno and Garner, Richard},
  journal = {Proceedings of the London Mathematical Society},
  volume = {102},
  number = {2},
  pages = {370--394},
  year = {2011}
}

@inproceedings{Lumsdaine2009,
  title = {Weak $\omega$-categories from intensional type theory},
  author = {Lumsdaine, Peter LeFanu},
  booktitle = {Typed Lambda Calculi and Applications},
  series = {Lecture Notes in Computer Science},
  volume = {5608},
  pages = {172--187},
  publisher = {Springer},
  address = {Berlin, Germany},
  year = {2009}
}

@incollection{Awodey2012,
  title = {Type Theory and Homotopy},
  author = {Awodey, Steve},
  booktitle = {Epistemology versus Ontology},
  series = {Logic, Epistemology, and the Unity of Science},
  volume = {27},
  editor = {Dybjer, Peter and Lindstr{"o}m, Sten and Palmgren, Erik and Sundholm, G{"o}ran},
  pages = {183--201},
  publisher = {Springer},
  address = {Dordrecht, Netherlands},
  year = {2012}
}

@book{MartinLoef1984,
  author = {Martin-L{"o}f, Per},
  title = {Intuitionistic Type Theory},
  publisher = {Bibliopolis},
  address = {Naples},
  year = {1984}
}

@phdthesis{Ramos2018Thesis,
  author = {Ramos, Arthur F.},
  title = {Explicit Computational Paths in Type Theory},
  school = {Universidade Federal de Pernambuco},
  year = {2018},
  address = {Recife, Brazil},
  url = {https://github.com/Arthur742Ramos/ComputationalPathsLean/blob/main/docs/thesis/TESE\%20Arthur\%20Freitas\%20Ramos.pdf},
  note = {Available at \url{https://github.com/Arthur742Ramos/ComputationalPathsLean/blob/main/docs/thesis/TESE\%20Arthur\%20Freitas\%20Ramos.pdf}}
}

@book{BaaderNipkow1998,
  title = {Term Rewriting and All That},
  author = {Baader, Franz and Nipkow, Tobias},
  publisher = {Cambridge University Press},
  address = {Cambridge, UK},
  year = {1998}
}

@incollection{KnuthBendix1970,
  title = {Simple Word Problems in Universal Algebras},
  author = {Knuth, Donald E. and Bendix, Peter B.},
  booktitle = {Computational Problems in Abstract Algebra},
  editor = {Leech, J.},
  publisher = {Pergamon},
  address = {Oxford, UK},
  pages = {263--297},
  year = {1970}
}

@article{Cohen2018Cubical,
  title = {Cubical Type Theory: A Constructive Interpretation of the Univalence Axiom},
  author = {Cohen, Cyril and Coquand, Thierry and Huber, Simon and M{"o}rtberg, Anders},
  journal = {Mathematical Structures in Computer Science},
  volume = {28},
  number = {7},
  pages = {1228--1269},
  year = {2018}
}

@misc{ComputationalPathsLean,
  author = {Ramos, Arthur F. and de Veras, Tiago M. L. and de Queiroz, Ruy J. G. B. and de Oliveira, Anjolina G.},
  title = {Computational Paths in Lean},
  howpublished = {\url{https://github.com/Arthur742Ramos/ComputationalPathsLean}},
  note = {Version accessed November 22, 2025},
  year = {2025}
}

@inproceedings{Ruy4,
  title = {Equality in labelled deductive systems and the functional interpretation of propositional equality},
  author = {de Queiroz, Ruy J. G. B.},
  booktitle = {Proceedings of the 9th Amsterdam Colloquium},
  year = {1994},
  pages = {547--565},
  address = {Amsterdam, Netherlands}
}

@article{leinster1,
  title = {A survey of definitions of n-category},
  author = {Leinster, Tom},
  journal = {Theory and Applications of Categories},
  volume = {10},
  pages = {1--70},
  year = {2002}
}

@article{batanin1,
  title = {Monoidal globular categories as a natural environment for the theory of weak n-categories},
  author = {Batanin, Michael A.},
  journal = {Advances in Mathematics},
  volume = {136},
  pages = {39--103},
  year = {1998}
}

@inproceedings{FinsterMimram2017,
  title = {A Type-Theoretical Definition of Weak $\omega$-Categories},
  author = {Finster, Eric and Mimram, Samuel},
  booktitle = {Proceedings of the 32nd Annual ACM/IEEE Symposium on Logic in Computer Science (LICS)},
  pages = {1--12},
  year = {2017},
  publisher = {IEEE},
  doi = {10.1109/LICS.2017.8005124},
  note = {Extended version: arXiv:1706.02866}
}

@article{BenjaminFinsterMimram2021,
  title = {Globular weak $\omega$-categories as models of a type theory},
  author = {Benjamin, Thibaut and Finster, Eric and Mimram, Samuel},
  journal = {arXiv preprint arXiv:2106.04475},
  year = {2021},
  note = {Proves the initiality conjecture for CaTT}
}

@inproceedings{LjungstromMortberg2023,
  title = {Formalizing $\pi_4(\mathbb{S}^3) \cong \mathbb{Z}/2\mathbb{Z}$ and Computing a Brunerie Number in Cubical Agda},
  author = {Ljungstr{\"o}m, Axel and M{\"o}rtberg, Anders},
  booktitle = {Proceedings of the 38th Annual ACM/IEEE Symposium on Logic in Computer Science (LICS)},
  year = {2023},
  publisher = {IEEE},
  doi = {10.1109/LICS56636.2023.10175833}
}

@article{AnnenkovCapriottiKrausSattler2023,
  title = {Two-level type theory and applications},
  author = {Annenkov, Danil and Capriotti, Paolo and Kraus, Nicolai and Sattler, Christian},
  journal = {Mathematical Structures in Computer Science},
  volume = {33},
  number = {8},
  pages = {688--743},
  year = {2023},
  publisher = {Cambridge University Press},
  doi = {10.1017/S0960129523000130}
}

@article{KrausVonRaumer2022,
  title = {A rewriting coherence theorem with applications in homotopy type theory},
  author = {Kraus, Nicolai and von Raumer, Jakob},
  journal = {Mathematical Structures in Computer Science},
  volume = {32},
  number = {7},
  pages = {982--1014},
  year = {2022},
  publisher = {Cambridge University Press},
  doi = {10.1017/S0960129522000329},
  note = {Shows how confluence and wellfoundedness yield coherence in higher-dimensional settings}
}

@inproceedings{Mimram2023Coherence,
  title = {Categorical Coherence from Term Rewriting Systems},
  author = {Mimram, Samuel},
  booktitle = {8th International Conference on Formal Structures for Computation and Deduction (FSCD 2023)},
  series = {Leibniz International Proceedings in Informatics},
  volume = {260},
  pages = {16:1--16:17},
  year = {2023},
  publisher = {Schloss Dagstuhl},
  doi = {10.4230/LIPIcs.FSCD.2023.16}
}

\end{document}